\documentclass[11pt]{IEEEtran}
\usepackage{amsfonts} 
\usepackage{amsmath} 
\usepackage{amssymb}

\newcommand{\N}{\mathbb{N}}
\newcommand{\Z}{\mathbb{Z}}
\newcommand{\Ss}{\mathcal{S}}
 
\newcommand{\A}{\alpha}
\newcommand{\Ll}{\lambda}
\newcommand{\F}{\mathbb{F}}
 
\newcommand{\rs}{\mathrm{rowsp}}
\newcommand{\Gr}[3]{\mathrm{G}(#2,#1^{#3})}
 
\newcommand{\vier}[4]{\left[\begin{array}{cc} #1 & #2 \\ #3
      & #4 \end{array}\right]}
\newcommand{\Mat}{\mathrm{Mat}}
\newcommand{\mat}[1]{\left(\begin{matrix}#1\end{matrix} \right)} 
\newcommand{\rank}{\mathrm{rank}}

\newtheorem{thm}{Theorem} 
\newtheorem{lem}[thm]{Lemma}
\newtheorem{co}[thm]{Corollary}

\newtheorem{rem}[thm]{Remark}

\newtheorem{defi}[thm]{Definition}

\begin{document}

\title{Spread Codes and Spread Decoding \\ in Network
  Coding}

\author{ Felice Manganiello, Elisa Gorla and Joachim
  Rosenthal\\
  Mathematics Institute\\
  University of Zurich\\
  Winterthurerstr 190\\
  CH-8057 Zurich,  Switzerland\\
  \large{\texttt{www.math.uzh.ch/aa}} \thanks{First and third
    author were partially supported by Swiss National Science
    Foundation under Grant no.\ 113251.  Second Author was supported
    by the Forschungskredit of the University of Zurich under Grant
    no.\ 57104101 and by the Swiss National Science
    Foundation under Grant no.\ 107887.}}

\date{}

\maketitle

\begin{abstract}
  In this paper we introduce the class of {\em Spread Codes} for
  the use in random network coding. Spread Codes are based on the
  construction of spreads in finite projective geometry. The
  major contribution of the paper is an efficient decoding
  algorithm of spread codes up to half the minimum distance.
\end{abstract}

\section{Introduction}
In~\cite{ko07} K\"otter and Kschischang develop a novel framework
for random network coding. In this framework information is
encoded in subspaces of a given ambient space over a finite
field. A natural metric is introduced where two subspaces are
`close to each other' as soon as their dimension of intersection is
large. This new framework poses new challenges to design new codes
with large distances and to come up with efficient decoding
algorithms. Several new papers have been written on the topic and
we mention~\cite{si07} and~\cite{mo07}.

In this paper we study the class of spreads from finite
projective geometry (see e.g.~\cite{hi98}) for possible use in
network coding theory.  A spread $\Ss$ is a partition of a vector
space by subspaces of a fixed dimension. Elements of
a spread are subspaces of a fixed vector space
$\F_q^n$ which pairwise only intersect in the origin. 
 The codewords derived in this way are all
subspaces of the same dimension. In other words the spread  $\Ss$
is a subset of the finite
Grassmannian $\Gr{\F_q}{k}{n}$ consisting of all $k$-dimensional
subspaces in $\F_q^n$. We will call the obtained code a {\em
  Spread code}.
Since two different elements of $\Ss$ only
intersect in the origin the spread code  $\Ss$ has maximal
possible distance among all subsets of $\Gr{\F_q}{k}{n}$. 

The paper is structured as follows.  In the next section we will
explain the construction of spreads and we derive some basic
properties. In Section 3 the main results of the paper are given. We
provide an efficient decoding algorithm for spread codes essentially
`up to half the minimum distance' with its complexity. The decoding
algorithm requires methods from linear algebra and the application of
the Euclidean algorithm.

\section{Algebraic Construction of a Spread Code}

Let $\F_q$ be the finite field with $q$ elements. We denote with
$\Gr{\F_q}{k}{n}$ the Grassmannian of all $k$-dimensional
subspaces of $\F_q^n$.  Following~\cite{ko07} we define a
distance function $d:\Gr{\F_q}{k}{n}\times \Gr{\F_q}{k}{n}
\rightarrow \Z_+$ through:
\begin{eqnarray}\label{dist}
d(A,B)&:=&\dim(A+B)-\dim(A\cap B)\\
\nonumber & = &\dim(A)+\dim(B)-2\dim(A\cap B).
\end{eqnarray}
It has been observed in~\cite{ko07} that $d(A,B)$ satisfies the axioms
of a metric on the finite Grassmannian $\Gr{\F_q}{k}{n}$.  A
constant-dimension code $\Ss\subset \Gr{\F_q}{k}{n}$ has maximal
possible minimum distance as long as the intersection of two different
codewords of $\Ss$ is trivial.  If two subspaces $A,B\subset \F_q^n$
intersect only in the zero vector then the corresponding subspaces of
projective space are non-intersecting. Based on this we will call
$A,B\subset \F_q^n$ nonintersecting subspaces as long as they
intersect only in the zero vector.

We want to construct an MDS-like code $\Ss\subset
\Gr{\F_q}{k}{n}$, i.e. code having maximum possible distance and
maximum number of elements. In order to do this we
need to restrict our $k,n\in \N$ to some particular cases. It is a well
known result that there exists an $\Ss\subset \Gr{\F_q}{k}{n}$ that
partitions $\F_q^n$ (i.e.  there is no vector in $\F_q^n$ which does not
lie in a subspace) and such that any two elements of $\Ss$ are
nonintersecting if and only if $k$ divides $n$. Those subsets are
called \emph{spreads} and this result can be found in \cite{hi98}.

Consider the case $n=rk$. Let also $p\in \F_q^n[x]$ be an irreducible
polynomial of degree $k$. If we denote with $P$ the $k\times k$
companion matrix of $p$ over $\F_q$, it follows that the
$\F_q$-algebra $\F_q[P]\subset \Mat_{k\times k}(\F_q)$ is isomorphic to
the finite field $\F_{q^k}$.  Denoting with $0_k,I_k\in \Mat_{k\times
  k}(\F_q)$ respectively the zero and the identity matrix and given the
above assumptions, we are ready to state the following theorem.

\begin{thm}\label{constr}
The collection of subspaces 
\begin{eqnarray*}
\Ss:=\bigcup_{i=1}^r & \left\{  \rs \left[ 0_k \ \cdots \ 0_k \ I_k \ A_{i+1} \
    \cdots \ A_r \right] \mid \right.\\ 
 & \left.A_{i+1},\dots,A_{r} \in\F_q[P]\right\} \subset \Gr{\F_q}{k}{n}
\end{eqnarray*} 
is a spread of $\F_q^n$.
\end{thm}

\begin{proof}
  The cardinality of $\Ss$ is exactly the maximum number of
  $k$-dimensional nonintersecting subspaces of $\F_q^n$, i.e.
  $\frac{q^n-1}{q^k-1}=q^{k(r-1)}+q^{k(r-2)}+\cdots+q+1$.

  It remains to be shown that any pair of subspaces in $\mathcal{S}$
  do only intersect trivially that is equivalent to showing that the
  $2k\times n$ matrix obtained putting together two matrices
  generating two different subspaces is full-rank.

  We have only two cases. The first where the matrices $I_k$ are not
  placed at the same column ``level''. In this case we can find a full-rank
  submatrix of the form\[\vier{I_k}{A}{0_k}{I_k}.\]   The second case
  is when matrices $I_k$ are at the same ``level''. There exists a
  submatrix of the form \[\vier{I_k}{A_1}{I_k}{A_2}\] where
  $A_1,A_2\in \F_q[P]$ and $A_1 \neq A_2$. It follows that the determinant
  of the above matrix is equal to $\det(A_1-A_2)$ and is nonzero since
  $A_1\neq A_2$.
\end{proof}

Is it possible to find a previous and less general version of this
theorem in \cite{cl07a}.   

\begin{defi}
  Let $p$ be an irreducible polynomial of degree $k$ over $\F_q$. A
  \emph{spread code} $\Ss$ is a subset of $\Gr{\F_q}{k}{n}$ constructed
  as in the previous theorem. Following the definition of \cite{ko07}
  a spread code is a $q$-ary code of type
  $[n,k,\log_q\left(\frac{q^n-1}{q^k-1}\right),2k]$.
\end{defi}

\begin{rem}
  Spread codes are related to the Reed-Solomon-like codes over
  Grasmannians presented in the paper \cite{ko07}. Following the
  notation of \cite{ko07}, let $l=k$ and $m=n-k$. From the
  construction of Theorem \ref{constr}, if follows that the subset of
  $\Ss$ with $i=1$ is a subcode of Reed-Solomon-like codes. Moreover,
  our costruction provides more codewords arising from the cases where
  $i>1$.
\end{rem}

There is an algebraic geometric way to view the spreads we just
introduced. For this identify the set of polynomials in $\F_q[x]$
having degree at most $k-1$ with the field $F_{q^k}$. Consider the
natural isomorphism
\begin{eqnarray*}
\varphi:\F_{q^k} &\rightarrow& \F_q[P]\\
f&\mapsto &f(P).
\end{eqnarray*}
This isomorphism induces the natural embedding
\begin{eqnarray*}
  \tilde{\varphi}:\Gr{\F_{q^k}}{l}{m}\rightarrow  \Gr{\F_q}{kl}{km}
\end{eqnarray*}
with

\begin{eqnarray*}
  \tilde{\varphi}\left(\rs\mat{f_{11} & \dots & f_{1m} \\ \vdots & & \vdots \\ f_{l1} &
    \dots &f_{lm}} \right) \\ \mbox{ \ \ \ \ \ }= \rs\mat{f_{11}(P) & \dots & f_{1m}(P)
    \\ \vdots & & \vdots \\f_{l1}(P) & \dots &f_{lm}(P)}. 
\end{eqnarray*} 

The following theorem is then not difficult to establish.

\begin{thm}
  If $\Ss\subset \Gr{\F_{q^k}}{l}{m}$ is a spread of $\F_{q^k}^m$ then 
  $\tilde{\varphi}(\Ss) \subset \Gr{\F_q}{kl}{km}$ is a spread of
  $\F_q^{km}$. 
\end{thm}

Clearly $\Gr{\F_{q^k}}{1}{r}$ is a spread itself and it therefore
follows that the subset defined in Theorem \ref{constr} is a spread of
$\F_q^n$ as well.
 
\section{Decoding Algorithm}

We will continue restricting our study to the case where $n=2k$ and
$k$ is odd. From now on we will consider fixed the irreducible
polynomial $p\in \F_q[x]$. 

In a first step we want to establish a simple algebraic criterion
which characterizes the spread code $\mathcal{S}\subset\Gr{\F_q}{k}{2k}$.
For this assume that $C_1,C_2 \in \Mat_{k\times k}(\F_q)$ are matrices
such that
\[C:=\mathrm{rowsp}[C_1\ C_2]\in \Gr{\F_q}{k}{2k}.\] If $C_1$ is not
invertible then $C\in\mathcal{S}$ if and only if $C_1=0_k$. If $C_1$ is
invertible then $C\in\mathcal{S}$ if and only if
$A:=(C_1)^{-1}C_2\in\F_q[P]$. 

We therefore establish a criterion which guarantees that a matrix $A$
is in $\F_q[P]$.  Let $\F_{q^k}$ be the splitting field of $p$ over
$\F_q$ and $S\in Gl_k(\F_{q^k})$ be an invertible matrix diagonalizing
the matrix $P$, i.e.
\[D:=SPS^{-1}=\left[ \begin{array}{ccccc}
    \Ll& & & & \\ 
     & \Ll^q & & & \\
     & & \ddots & & \\
     & & & & \Ll^{q^{k-1}} \end{array} \right]\] 
where $\Ll\in \F_{q^k}$ is a root of $p$.

\begin{lem}
  Let $A\in \Mat_{k\times k}(\F_q)$. Then $A\in\F_q[P]$ if and only if
  $AP=PA$.  
\end{lem}

\begin{proof}
  If $A\in\F_q[P]$ then clearly $AP=PA$. Assume now $AP=PA$ and
  $SPS^{-1}=D$. Since the eigenvalues of $P$ are pairwise
  different and $D(SAS^{-1})=(SAS^{-1})D$ it follows that
  $SAS^{-1}$ is a diagonal matrix as well with diagonal entries
  in $\F_{q^k}$. Let $\{1,\gamma,\ldots,\gamma^{k-1}\}$ be a
  basis of $\F_{q^k}$ over $\F_{q}$. One has an expansion:
  $$
  SAS^{-1}=\sum_{i=0}^{k-1}c_iD^i=
  \sum_{i=0}^{k-1}\sum_{j=0}^{k-1}c_{i,j}\gamma^{j}D^i$$ 
  with $c_i\in\F_{q^k}$ and $c_{i,j}\in \F_q$.

  Equivalently we have:
  $$
  A=\sum_{j=0}^{k-1}\left(\sum_{i=0}^{k-1}c_{i,j}P^i\right)\gamma^{j}.
  $$
  It follows that $A=\sum_{i=0}^{k-1}c_{i,0}P^i$ and $A\in
  \F_q[P]$.
\end{proof}

The following gives an algebraic criterion for checking when a
subspace is a codeword.

\begin{co}\label{diag}
  The subspace $\rs[I_k\ A]\in \Gr{\F_q}{k}{2k}$ is a codeword of $\Ss$
  if and only if $SAS^{-1}$ is a diagonal matrix.
\end{co}

We state now the unique decoding problem. Assume $C:=\rs[C_1\ C_2]\in
\mathcal{S}$ was sent and $R:=\rs[R_1\ R_2]\in \Gr{\F_q}{k}{2k}$ was
received.  If
\begin{eqnarray}\label{1dec} 
\dim (C\cap R) \geq \frac{k+1}{2}
\end{eqnarray}
then unique decoding is possible. In the sequel we will consider the
received subspace $R\in \Gr{\F_q}{k}{2k}$ such that there
exists a codeword $C\in \Ss$ such that \eqref{1dec} holds.

\subsection{Case $R_1$ not invertible.}

Let $R$ and $C$ be subspaces satisfying the condition \eqref{1dec}. The
goal of this subsection is to analyze the behavior of the decoding
problem when $R_1$ is not invertible.

This situation splits in two different ones. The first one is when
$0\leq \rank(R_1) \leq \frac{k-1}{2}$. The closest codeword in this
case is only the subspace $\rs[0_k \ I_k]$.

The second case is characterized by $\frac{k+1}{2}\leq \rank(R_1)\leq
k-1$. With the following lemma we bring back the decoding problem of
the subspace $R$ to the one of a subspace $\tilde{R}$ close related to $R$
and lying in the same ball with center in the codeword $C$. 

\begin{lem}
  Let $R\in \Gr{\F_q}{k}{2k}$ such that $\frac{k+1}{2}\leq
  \rank(R_1)\leq k-1$ and $C\in \Ss$ such that \eqref{1dec}
  holds. Then there exists a subspace $\tilde{R}:=\rs[\tilde{R_1} \
  \tilde{R_2}]\in \Gr{\F_q}{k}{2k}$ satisfying:
  \begin{itemize}
  \item $\tilde{R_1}$ is invertible,
  \item $\dim(R\cap\tilde{R})=\rank(R_1)$, and
  \item $\dim(C\cap \tilde{R})\geq \frac{k+1}{2}$.
  \end{itemize}
\end{lem}

\begin{proof}
  Let $t:=\rank(R_1)$. Row reducing the matrix $[R_1 \
  R_2]$ we obtain the matrix $\vier{\bar{R_1}}{\bar{R_2}}{0}{E}$ where
  $\bar{R_1}, \bar{R_2}\in \Mat_{t\times k}(\F_q)$ with $R_1$ fullrank and
  $0,E\in \Mat_{k-t\times k}(\F_q)$ where $0$ is the zero matrix.

  Since $\rs[0 \ E]\subset \rs[0_k \ I_k]$ we deduce that $\dim(C\cap
  \rs[0 \ E])=0$. It follows immediately that \[\dim(C\cap
  \rs[\bar{R_1} \ \bar{R_2}])=\dim(C\cap \tilde{R})\geq
  \frac{k+1}{2}.\] The matrix representing the subspace $\tilde{R}$ can
  then be constructed as follows:
\begin{itemize}
\item $\tilde{R_1}$ is the completion of the matrix $\bar{R_1}$ to an
  invertible matrix, and
\item $\tilde{R_2}$ is the completion of the $\bar{R_2}$ to a
  $k$-square matrix by adding rows of zeros.  
\end{itemize}
\end{proof}

\begin{co}
  The solution to the unique decoding problem for both subspaces $R$
  and $\tilde{R}$ consists of the same codeword $C\in \Ss$.
\end{co}

\subsection{Case $R_1$ invertible.}

We can now construct an algorithm for the unique decoding problem of
subspaces with $R_1$ invertible.

\begin{thm} \label{decomp} Let $R:=\mathrm{rowsp}[R_1\ R_2]\in
  \Gr{\F_q}{k}{2k}$ a subspace with $R_1$ invertible. Then there exists
  a unique matrix $A\in \F_q[P]$ and a unique matrix $N\in
  \Mat_{k\times k}(\F_q)$ of rank at most $\frac{k-1}{2}$ such that
\[R_1^{-1}R_2=A +N.\]
  In this case $\mathrm{rowsp}[I_k\ A]$ is the closest codeword to $R$
  in the distance \eqref{dist}.
\end{thm}

\begin{proof}
  The uniqueness follows from the distance properties of the code.
  Assume $\mathrm{rowsp}[I_k\ A]$ be the closest codeword to $R$. Since
  \[\mathrm{rowsp}\vier{I_k}{A}{R_1}{R_2}=
  \mathrm{rowsp}\vier{I_k}{A}{0_k}{R_1^{-1}R_2-A}\] has dimension at
  most $2k-\frac{k+1}{2}=k+\frac{k-1}{2}$ it follows that the matrix
  $N:=R_1^{-1}R_2-A$ has rank at most $\frac{k-1}{2}$.
\end{proof}

\begin{co}\label{unique}
  Let $R:=\mathrm{rowsp}[R_1\ R_2]\in \Gr{\F_q}{k}{2k}$ a subspace with
  $R_1$ invertible. Let $Y:=S(R_1^{-1}R_2)S^{-1}$. Then there is a
  unique polynomial $f\in\F_q[x]$ with $\deg f<k$ such
  that $Y-f(D)$ has rank at most $\frac{k-1}{2}$.
\end{co}

\begin{proof}
The existence follows directly from the last theorem. Concerning
the uniqueness assume that $Y=f_1(D)+N_1=f_2(D)+N_2$. It then
follows that 
$$
R_1^{-1}R_2=f_1(P)+S^{-1}N_1S=f_2(P)+S^{-1}N_2S
$$
and because of the uniqueness part of Theorem~\ref{decomp} the
result follows.
\end{proof}

The algorithm extrapolates the evaluations of the polynomial $f\in
\F_q[x]$ from the matrix $Y-f(D)$.  Once the polynomial $f\in \F_q[x]$
is found, its evaluation at $P$ gives us the matrix $A\in \F_q[P]$
such that $\rs\left[I_k \ A\right]$ is the codeword closest to $R$.
Notice that the coefficients of $f$ are exactly the coefficients of
the expression of $f(\Ll)$ in the basis $\{1,\Ll,\dots,\Ll^{k-1}\}$ of
$\F_{q^k}$ over $\F_q$.

The following two remarks from finite field theory (see \cite{li94})
will be important. First, given any $f\in \F_q[x]$ and any $\mu\in
\F_{q^k}$, then $f(\mu^q)=f(\mu)^q$. Second, given a finite field
$\F_q$ with $q$ elements it holds \[x^q-x=\prod_{\A\in \F_q}(x-\A).\]

We outline now the complete decoding algorithm. 

Let $R:=\rs\left[R_1 \ R_2\right]$ be the received subspace satisfying
condition \eqref{1dec}. Assume that $R_1$ is invertible. Compute
$Y:=S(R_1^{-1}R_2)S^{-1}$. If the matrix $Y$ is diagonal, then $R$ is
already a codeword of $\Ss$ by Corollary \ref{diag}.

Otherwise the matrix $Y-f(D)$ is of the form 
 \begin{eqnarray*}
\mat{y_{1,1}-f(\Ll) & y_{1,2} & \cdots & y_{1,k} \\
   y_{2,1} & y_{2,2}-f(\Ll^q) & \cdots & y_{2,k} \\
   \vdots & \vdots & \ddots & \vdots \\
   y_{k,1} & y_{k,2} & \cdots &
   y_{k,k}-f(\Ll^{q^{k-1}})} = \\\mat{y_{1,1}-f(\Ll) & y_{1,2} & \cdots &
   y_{1,k} \\
   y_{2,1} & y_{2,2}-f(\Ll)^q & \cdots & y_{2,k} \\
   \vdots & \vdots & \ddots & \vdots \\
   y_{k,1} & y_{k,2} & \cdots & y_{k,k}-f(\Ll)^{q^{k-1}}}
\end{eqnarray*} 
where some entries off of the diagonal are nonzero.  Denote by $X$ the
matrix obtained from $Y-f(D)$ by substituting $x$ for $f(\Ll)$. By
Corollary \ref{unique} there exists a unique value for $x\in \F_{q^k}$
(namely $x=f(\Ll)$) such that $\mathrm{rank} (X) \leq
\frac{k-1}{2}$. The decoding problem reduces to finding such a value.

The condition on the rank is equivalent to having all minors of size
$\frac{k+1}{2}$ of the matrix $X$ being zero. This gives us a system
of univariate equations which apriori may be hard to solve.  However
since the system has a unique solution, every minor is divisible by
$(x-f(\Ll))$.

Hence in order to find $f(\Ll)$ it suffices to compute the gcd of the
field equation $x^{q^k}-x$ with enough equations from our system.
More precisely we look for a nonzero minor of size $\frac{k-1}{2}$
which does not involve any diagonal entry. If no such minor exists,
then look for a nonzero minor of smaller size which again does not
involve any diagonal entry. Let $t$ be the size of the minor.
Complete the corresponding size $t$ submatrix to a submatrix of $X$ of
size $\frac{k+1}{2}$. Notice that this can be done by adding
$\frac{k+1}{2}-t$ rows and columns with the same index. The
determinant of this submatrix is a nonzero polynomial $m\in
\F_{q^k}[x]$ which has $f(\Ll)$ as a root.

Apply the Euclidean Algorithm in order to compute
\[g:=\gcd(x^{q^k}-x,m).\] If the degree of $g$ is small, compute its
roots and substitute them in $X$ in order to find $f(\Ll)$.

Otherwise compute another minor in the same way as for the previous
one. Proceed by computing the gcd of this polynomial with $g$. The
algorithm ends once it finds $f(\Ll)$.

\subsection{Complexity} 

The overall complexity of the algorithm is dominated by the Euclidean
Algorithm. In the worst case scenario, i.e. when the maximal nonzero
minor off diagonal has size 1, the algorithm's complexity is
$\mathcal{O}(q^{k\log_2 3}\log q^k)$ in $\F_{q^k}$.

The complexity could be drastically decreased by the following
conjecture: for every error matrix $N\in \Mat_{k\times k}(\F_q)$ of
rank $t\leq\frac{k-1}{2}$ there exists a nonzero minor of size
$t$ of the matrix $X$ which does not involve any diagonal entry.

Consider now such a nonzero minor of X and extend the related
submatrix adding one row and one column with the same index. The
determinant of this submatrix leads to an equation of the type
$x^{q^i}=\alpha$ with $\alpha \in \F_q$. Raising both sides of the
equation to the $q^{k-i}$-th power and using the field equation of
$\F_{q^k}$ we get: $x=\alpha^{q^{k-i}}.$ Using the Repeated Squaring
Algorithm for computing powers in $\F_{q^k}$, the complexity of the
decoding algorithm decreases to $\mathcal{O}(\log
q^{k-i})=\mathcal{O}(k-i)$ operations in $\F_{q^k}$.

A reference for efficient algorithms is \cite{ga03}. In particular see
Section 4.3 for the Repeated Squaring Algorithm, Section 11.1 for
performing the Euclidean Algorithm, Chapter 14 for factoring
univariate polynomials and Section 25.5 for computing determinants.

\subsection{Non-perfectness of a Spread Code}

Spreads are perfect in the sense that every nonzero vector of
$\F_q^n$ is in one and only one subspace of the spread. 

In coding theory a code is perfect if the total ambient space is
covered with the balls centered in the codewords and having radius
half the minimum distance. It arises the question if spread codes are
perfect in this sense. The answer turns out to be negative in general
and this result can be found in \cite{ma95c}.
 
\section*{Acknowledgments}

We would like to thank Joan Josep Climent, Felix Fontein,
Ver\'onica Requena and Jens Zumbr\"agel for many helpful
discussions during the preparation of this paper.

\def\cprime{$'$} \def\cprime{$'$}

\end{document}